\documentclass{IEEEtran}
\usepackage{cite}
\usepackage{graphicx}
\usepackage{psfrag}
\usepackage{subfigure}
\usepackage{url}
\usepackage{amsmath}
\usepackage{array}
\usepackage{amssymb}
\usepackage{amsfonts}
\newtheorem{proposition}{Proposition}
\newtheorem{lemma}{Lemma}

\newtheorem{corollary}{Corollary}
\newtheorem{remark}{Remark}
\newtheorem{definition}{Definition}
\newtheorem{theorem}{Theorem}
\newtheorem{example}{Example}

\title{Network error correction for unit-delay, memory-free networks using convolutional codes}
\author{K.~Prasad and B. Sundar Rajan, Senior member, IEEE}
\date{\today}
\begin{document}
\maketitle
\thispagestyle{empty}
\begin{abstract}
A single source network is said to be \textit{memory-free} if all of the internal nodes (those except the source and the sinks) do not employ memory but merely send linear combinations of the symbols received at their incoming edges on their outgoing edges. In this work, we introduce network-error correction for single source, acyclic, unit-delay, memory-free networks with coherent network coding for multicast. A convolutional code is designed at the source based on the network code in order to correct network-errors that correspond to any of a given set of error patterns, as long as consecutive errors are separated by a certain interval which depends on the convolutional code selected. Bounds on this interval and the field size required for constructing the convolutional code with the required free distance are also obtained. We illustrate the performance of convolutional network error correcting codes (CNECCs) designed for the unit-delay networks using simulations of CNECCs on an example network under a probabilistic error model.
\end{abstract}
\section{Introduction}
\label{sec1}
Network coding was introduced in \cite{ACLY} as a means to improve the rate of transmission in networks, and often achieve capacity in the case of single source networks. Linear network coding was introduced in \cite{CLY}. An algebraic formulation of network coding was discussed in \cite{KoM} for both instantaneous networks and networks with delays. 

Network error correction, which involved a trade-off between the rate of transmission and the number of correctable network-edge errors, was introduced in \cite{YeC} as an extension of classical error correction to a general network setting. Along with subsequent works \cite{Zha} and \cite{YaY}, this generalized the classical notions of the Hamming weight, Hamming distance, minimum distance and various classical error control coding bounds to their network counterparts. In all of these works, it is assumed that the sinks and the source know the network topology and the network code, which is referred to as \textit{coherent network coding}. Network error correcting codes were also developed for \textit{non-coherent (channel oblivious) network coding} in \cite{KoK},\cite{SKK} and \cite{EtS}. Network error correction under probabilistic error settings has been studied in \cite{SKK2}. Most recently, multishot subspace codes were introduced in \cite{NoU} for the subspace channel \cite{KoK} based on block-coded modulation. 

A set of code symbols generated at the source at any particular time instant is called a \textit{generation} of code symbols. So far, network error correcting schemes have been studied only for acyclic \textit{instantaneous} (delay-free) networks in which each node could take a linear combination of symbols of only the same generation. 

Convolutional network codes were discussed in \cite{ErF,CLYZ,LiY} and a connection between network coding and convolutional coding was analyzed in \cite{FrS}. Convolutional network error correcting codes (which we shall henceforth refer to as CNECCs) have been employed for network error correction in instantaneous networks in \cite{PrR}. 

A \textit{network use} \cite{PrR} is a single usage of all the edges of the network to multicast utmost min-cut number of symbols to each of the sinks. An \textit{error pattern} is a subset of the set of edges of the network which are in error. It was shown in \cite{PrR} that any network error which has its error pattern amongst a given set of error patterns can be corrected by a proper choice of a convolutional code at the source, as long as consecutive network errors are separated by a certain number of network uses. Bounds were derived on the field size for the construction of such CNECCs, and on the minimum separation in network uses required between any two network errors for them to be correctable.

\textit{Unit-delay networks} \cite{CLYZ} are those in which every link between two nodes has a single unit of delay associated with it. In this work, we generalize the approach of \cite{PrR} to the case of network error correction for acyclic, unit-delay, memory-free networks. We consider single source acyclic, unit-delay, memory-free networks where coherent network coding (for the purpose of multicasting information to a set of sinks) has been implemented and thereby address the following problem.  

\textit{Given an acyclic, unit-delay, single source, memory-free network with a linear multicast network code, and a set of error patterns $\Phi$, how to design a convolutional code at the source which will correct network errors corresponding to the error patterns in $\Phi$, as long as consecutive errors are separated by a certain number of network uses?}

The main contributions of this paper are as follows.
\begin{itemize}
\item Network error correcting codes for unit-delay, memory-free networks are discussed for the first time. 
\item A convolutional code construction for the given acyclic, unit-delay, memory-free network that corrects a given pattern of network errors (provided that the occurrence of consecutive errors is separated by certain number of network uses) is given. For the same network, if the network code is changed, then the convolutional code obtained through our construction algorithm may also change. Several results of this paper can be treated as a generalization of those in \cite{PrR}. 
\item We derive a bound on the minimum field size required for the construction of CNECCs for unit-delay networks with the required minimum distance, following a similar approach as in \cite{PrR}. 
\item We also derive a bound on the minimum number of network uses that two error events must be separated by in order that they get corrected. 
\item We also introduce \textit{processing functions} at the sinks in order to address the realizability issues that arise in the decoding of CNECCs for unit-delay networks.
\item We show that the unit-delay network demands a CNECC whose free distance should be at least as much as that of any CNECC for the corresponding instantaneous network to correct the same number of network errors.
\item Using a probabilistic error model on a modified butterfly unit-delay memory-free network, we use simulations to study the performance of different CNECCs.  
\item Towards achieving convolutional network error correction, we address the issue of network coding for an acyclic, unit-delay, memory-free network. As a by-product, we prove that an $n$-dimensional linear network code (a set of local kernels at the nodes) for an acyclic, instantaneous network continues to be an $n$-dimensional linear network code (i.e the dimension does not reduce) for the same acyclic network, however being of unit-delay and memory-free nature.
\end{itemize}

The rest of the paper is organized as follows. 
In Section \ref{sec3}, we discuss the general network coding set-up and network errors. In Section \ref{sec4}, we give a construction for an input convolutional code for the given acyclic, unit-delay, memory-free network which shall correct errors corresponding to a given set of error patterns and also derive some bounds on the field size and minimum separation in network uses between two correctable network errors. In Section \ref{sec5}, we give some examples for this construction. In Section \ref{sec6} we provide a comparison between CNECCs for instantaneous networks \cite{PrR} and those for unit-delay, memory-free networks of this paper. In Section \ref{sec7}, we discuss the results of simulations of different CNECCs run on a modified butterfly network assuming a probabilistic model on edge errors in the network. We conclude this paper in Section \ref{sec8} with some remarks and some directions for further research.   

\section{Problem Formulation - CNECCs for unit-delay, memory-free networks}
\label{sec3}
\subsection{Network model}
We consider acyclic networks with delays in this paper, the model for which is as in  \cite{KoM}, \cite{CLYZ}. An acyclic network can be represented as an acyclic directed multi-graph (a graph that can have parallel edges between nodes) ${\cal G}$ = ($\cal V,\cal E$) where $\cal V$ is the set of all vertices and $\cal E$ is the set of all edges in the network. 

We assume that every edge in the directed multi-graph representing the network has unit \emph{capacity} (can carry utmost one symbol from $\mathbb{F}_q$). Network links with capacities greater than unit are modeled as parallel edges. The network has delays, i.e, every edge in the directed graph representing the input has a unit delay associated with it, represented by the parameter~$z$. Such networks are known as \textit{unit-delay networks}. Those network links with delays greater than unit are modeled as serially concatenated edges in the directed multi-graph. The nodes of the network may receive information of different generations on their incoming edges at every time instant. We assume that the internal nodes are memory-free and merely transmit a linear combination of the incoming symbols on their outgoing edges. 

Let $s\in\cal V$ be the source node and $\cal T$ be the set of all receivers. Let $n_{_T}$ be the unicast capacity for a sink node  $T\in{\cal T}$ i.e the maximum number of edge-disjoint paths from $s$ to $T$. Then 
\[
n = \min_{T\in{\cal T}}n_{_T}
\]
is the max-flow min-cut capacity of the multicast connection. 

\subsection{Network code}	
We follow  \cite{KoM} in describing the network code. For each node $v\in{\cal V}$, let the set of all incoming edges be denoted by $\Gamma_I(v)$. Then $|\Gamma_I(v)|=\delta_I(v)$ is the in-degree of $v$. Similarly the set of all outgoing edges is defined by $\Gamma_O(v)$, and the out-degree of the node $v$ is given by $|\Gamma_O(v)|=\delta_O(v)$.  For any $e \in {\cal E}$ and $v \in {\cal V}$, let $head(e)=v$, if $v$ is such that $e \in \Gamma_I(v)$. Similarly, let $tail(e)=v$, if $v$ is such that $e \in \Gamma_O(v)$.  We will assume an ancestral ordering on ${\cal E}$ of the acyclic graph ${\cal G}$.

The network code can be defined by the local kernel matrices of size $\delta_I(v)\times\delta_O(v)$ for each node $v\in{\cal V}$ with entries from $\mathbb{F}_q$. The global encoding kernels for each edge can be recursively calculated from these local kernels.

The network transfer matrix, which governs the input-output relationship in the network, is defined as given in \cite{KoM} for an $n$ dimensional network code. Towards this end, the matrices $A$,$K$,and $B^T$(for every sink $T\in {\cal T}$ are defined as follows.

The entries of the $n \times |{\cal E}|$ matrix $A$ are defined as
\[
A_{i,j}=\left\{
\begin{array}{cc}
\alpha_{i,e_j} & \text{   if } e_j \in \Gamma_{O}(s)\\
0  & \text{ otherwise}
\end{array}
\right. 
\]
where $\alpha_{i,e_j} \in \mathbb{F}_q$ is the local encoding kernel coefficient at the source coupling input $i$ with edge $e_j \in \Gamma_O(s)$.\newline 

The entries of the $|{\cal E}| \times |{\cal E}|$ matrix $K$ are defined as
\[
K_{i,j}=\left\{
\begin{array}{cc}
\beta_{i,j} & \text{   if } head(e_i) = tail(e_j) \\
0  & \text{ otherwise} 
\end{array}
\right. 
\]
where the set of $\beta_{i,j} \in \mathbb{F}_q$ is the local encoding kernel coefficient  between $e_i$ and $e_j$, at the node $v=head(e_i) = tail(e_j)$.\newline

For every sink $T \in {\cal T}$, the entries of the $|{\cal E}| \times n$ matrix $B^T$ are defined as  
\[
B^T_{i,j}=\left\{
\begin{array}{cc}
\epsilon_{e_j,i} & \text{   if } e_j \in \Gamma_{I}(T)\\
0  & \text{ otherwise} 
\end{array} 
\right. 
 \]
where all $\epsilon_{e_j,i} \in \mathbb{F}_q$.

For unit-delay, memory-free networks, we have 
\begin{eqnarray*}
F(z): = (I-zK)^{-1} 
\end{eqnarray*}
where $I$ is the $|{\cal E}| \times |{\cal E}|$ identity matrix.  Now we have the following definition.
\begin{definition}[\cite{KoM}]
\label{nettransmatrix}
\textit{The network transfer matrix}, $M_{T}(z)$, corresponding to a sink node ${T} \in \cal T$ is a full rank (over $\mathbb{F}_q(z)$) $n \times n$ matrix defined as 
\[
M_{T}(z):=AF(z)B^{T}=AF_{T}(z).
\] 
\end{definition} 	

With an $n$-dimensional network code, the input and the output of the network are $n$-tuples of elements from $\mathbb{F}_q[[z]].$ Definition \ref{nettransmatrix} implies that if $\boldsymbol{x}(z) \in \mathbb{F}_q^n[[z]]$ is the input to the unit-delay, memory-free network, then at any particular sink $T \in \cal T$, we have the output, $\boldsymbol{y}(z) \in \mathbb{F}_q^n[[z]]$, to be $\boldsymbol{y}(z) = \boldsymbol{x}(z)M_T(z).$
\subsection{CNECCs for single source, unit-delay, memory-free networks}
A primer on the basics of convolutional codes can be found in Appendix \ref{app1}. Assuming that an $n$-dimensional linear network code multicast has been implemented in the given single source unit-delay, memory-free network, we extend the definitions of the input and output convolutional codes of CNECCs for  instantaneous networks from \cite{PrR} to the unit-delay, memory-free case.
\begin{definition}
An \textit{input convolutional code}, ${\cal C}_s$, corresponding to an acyclic, unit-delay, memory-free network is a convolutional code of rate $~k/n (k < n)$ with a \textit{input generator matrix }$G_{I}(z)$ implemented at the source of the network.
\end{definition}
\begin{definition}
The \textit{output convolutional code} ${\cal C}_T$, corresponding to a sink node ${T} \in \cal T$ in the acyclic, unit-delay, memory-free network is the $~k/n (k < n)$ convolutional code generated by the \textit{output generator matrix} $G_{O,{T}}(z)$ which is given by
\[
G_{O,{T}}(z) = G_I(z)M_{T}(z)
\]
with $M_T(z)$ being the full rank network transfer matrix corresponding to an $n$-dimensional network code.
\end{definition}
\begin{example}
\begin{figure}[htbp]
\centering
\includegraphics[totalheight=2.2in,width=3.5in]{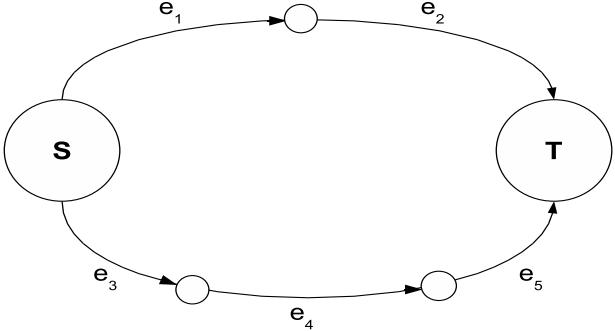}
\caption{A simple min-cut $2$ network with one source and one sink}	
\label{fig:simplenetwork}	
\end{figure}
Consider the single source, single sink network as shown in Fig.\ref{fig:simplenetwork}. Let the field under consideration be $\mathbb{F}_2.$ The local kernels at the intermediate node are unity. Therefore the network transfer matrix at the sink is (assuming the given ancestral ordering)
\[
M_T(z) = 
\left[
\begin{array}{cc}
z & 0 \\
0 & z^2
\end{array}
\right]
\]
Suppose we choose the input convolutional code ${\cal C}_s$ to be generated by the matrix
\[
G_I(z)=\left[1+z^2~~1+z+z^2\right].
\]
Then the output convolutional code ${\cal C}_T$ is generated by
\[
G_{O,T}(z)=\left[z+z^3~~z^2+z^3+z^4\right].
\]
\end{example}
\subsection{Network errors}

Observing a `snap-shot' of the network at any particular time instant, we define the following terms. An \textit{error pattern} $\rho,$ as stated previously, is a subset of ${\cal E}$ which indicates the edges of the network in error. An \textit{error vector} $\boldsymbol{w}$ is a $1\times |{\cal E}|$ vector which indicates the error occurred at each edge. An error vector is said to match an error pattern $(\text{i.e }\boldsymbol{w} \in \rho)$ if all non-zero components of $\boldsymbol{w}$ occur only on the edges in $\rho$. An \textit{error pattern set} $\Phi$ is a collection of subsets of ${\cal E}$, each of which is an error pattern. 

Let $\boldsymbol{x}(z) \in \mathbb{F}_q^n[[z]]$ be the input to the network , and $\boldsymbol{w} \in F_q^{|{\cal E}|}$ be the error vector corresponding to the network errors that occurred at any time instant $i$ ($i \in \mathbb{Z}_0^+$, referenced from the first input time instant). Then, the output, $\boldsymbol{y}(z) \in \mathbb{F}_q^n[[z]]$ at any particular sink $T \in \cal T$ can be expressed as 
\[
\boldsymbol{y}(z) = \boldsymbol{x}(z)M_T(z) + \boldsymbol{w}z^iF_T(z).
\]
In case there are a number of errors at a number of time instants, we have the formulation as
\[
\boldsymbol{y}(z) = \boldsymbol{x}(z)M_T(z)	 + \boldsymbol{w}(z)F_T(z)
\]
wherein every monomial of $\boldsymbol{w}(z) \in \mathbb{F}_q^{|\cal E|}[[z]]$ of the form $\boldsymbol{w}_iz^i$ incorporates the error vector $\boldsymbol{w}_i \in F_q^{|{\cal E}|}$ occurring at the time instant $i.$
\section{CNECCs for unit-delay, memory-free networks - Code Construction and Capability}
\label{sec4}
\subsection{Network code for acyclic unit-delay memory-free networks}
In Section \ref{construction}, we give a construction of a CNECC for a given acyclic, unit-delay, memory-free network. Towards that end, we first address the problem of constructing network codes for acyclic, unit-delay, memory-free networks. Although network code constructions have been given for acyclic instantaneous networks \cite{JSCEEJT}, the problem of constructing network codes for acyclic, unit-delay, memory-free networks is not directly addressed. The following lemma shows that solving an $n$-dimensional network code design problem for an acyclic, unit-delay, memory-free network is equivalent to solving that of the corresponding acyclic instantaneous network with the same number of dimensions. 
\begin{lemma}
Let ${\cal G}({\cal V},{\cal E})$ be a single source acyclic, unit-delay, memory-free network, and ${\cal G}_{inst}$ be the corresponding instantaneous network (i.e with the same graph as that of ${\cal G}$, but no delay associated with the edges). Let ${\cal N}$ be the set of all $\delta_I(v) \times \delta_O(v)$ matrices $\forall$ $v \in \cal V$, i.e, the set of local encoding kernel matrices at each node, describing an $n'$-dimensional network code (over $\mathbb{F}_q$) for ${\cal G}_{inst}$ ($n' \leq$ min-cut of the source-sink connections in ${\cal G}_{inst}$). Then the network code described by $\cal N$ continues to be an $n'$-dimensional network code (over $\mathbb{F}_q(z)$) for the unit-delay, memory-free network ${\cal G}.$
\end{lemma}
\begin{proof}
Let $M_T$ be the $n' \times n'$ network transfer matrix of any particular sink node $T \in \cal T$ in ${\cal G}_{inst}$, and $M_T(z)$ be the $n' \times n'$ network transfer matrix of the same sink $T$ in $\cal G.$ We first note that the matrix $M_T$ can be obtained from $M_T(z)$ by substituting $z=z^0=1$, i.e, 
\[
M_T = M_T(z)|_{z=1}.
\]
Given that $M_T$ is full rank over $\mathbb{F}_q$, we will prove that $M_T(z)$ is full rank over $\mathbb{F}_q(z)$ by contradiction.

Suppose that $M_T(z)$ was not full rank over $\mathbb{F}_q(z)$, then we will have 
\begin{equation}
\label{eqn1}
\sum_{i=1}^{i=n'-1}\frac{a_i(z)}{b_i(z)}\boldsymbol{m}_i(z) = \boldsymbol{m}_{n'}(z)
\end{equation}
where $\boldsymbol{m}_i(z)$ is the $i^{th}$ row of $M_T(z)$ and $a_i(z), b_i(z) \in \mathbb{F}_q[z]$ $\forall$ $i = 1,2,..,n'$ are such that $b_i(z) \neq 0, a_i(z) \neq 0$  for at least one $i$, and $gcd(a_i(z), b_i(z))=1, ~\forall~ i.$

We have the following two cases

\textit{Case 1}: $b_i(z)|_{z=1} \neq 0$ $\forall i.$

Substituting $z = 1$ in (\ref{eqn1}), we have 
\begin{equation}
\label{eqn3}
\sum_{i=1}^{i=n'-1}\frac{a_i}{b_i}\boldsymbol{m}_i = \boldsymbol{m}_{n'}
\end{equation}
where $a_i=a_i(z)|_{z=1}, b_i=b_i(z)|_{z=1}$ and $\boldsymbol{m}_i=\boldsymbol{m}_i(z)|_{z=1}$ is the $i^{th}$ row of $M_T.$ 

Clearly $\boldsymbol{m}_{n'} \neq \boldsymbol{0}$ since $M_T$ is full rank, and hence the left hand side of (\ref{eqn3}) can't be zero. Therefore some non-zero linear combination of the first ${n'}-1$ rows of $M_T$ is equal to its ${n'}^{^{th}}$ row, which contradicts the given statement that $M_T$ is full rank over $\mathbb{F}_q.$ Therefore $M_T(z)$ must be full rank over $\mathbb{F}_q(z).$

\textit{Case 2}:  $b_i(z)|_{z=1} = 0$ for at least one $i.$

Let ${\cal I}' \subseteq \left\{1,2,...,{n'}\right\}$ such that $(z-1)^{p'}|b_i(z)$ for some positive integer $p'.$ Let $p$ be an integer such that 
\[
p = \max_{i \in {\cal I}'}{p'}
\]
Now, from (\ref{eqn1}) we haven
\begin{equation}
\label{eqn2}
\sum_{i=1}^{i={n'}-1}(z-1)^{p}\frac{a_i(z)}{b_i(z)}\boldsymbol{m}_i(z) =(z-1)^{p} \boldsymbol{m}_{n'}(z)
\end{equation}
Let ${\cal I} \subseteq \left\{1,2,..,{n'}\right\}$ such that $(z-1)^{p}|b_i(z)$ $\forall$ $i \in \cal I.$ Then we must have that $(z-1) \nmid a_i(z)$ $\forall$ $i \in \cal I,$ since $gcd(a_i(z),b_i(z))=1.$ Also, let $b_i'(z) = b_i(z)/(z-1)^{p} \in \mathbb{F}_q[z]$ $\forall$ $i \in \cal I.$ Hence we have 
\[
\left((z-1)^{p} \frac{a_i(z)}{b_i(z)}\right)|_{z=1}=\left(\frac{a_i(z)}{b_i'(z)}\right)|_{z=1} = \frac{a_i}{b_i'} \in \mathbb{F}_q\backslash\left\{0\right\}.
\]
where $b_i' = b_i'(z)|_{z=1}\in \mathbb{F}_q\backslash\left\{0\right\}$, since ${(z-1)}\nmid{b_i'(z).}$ Substituting $z=1$ in (\ref{eqn2}), we have 
\[
\sum_{i\in {\cal I}} \frac{a_i}{b_i'}\boldsymbol{m}_i = \boldsymbol{0}
\]
i.e, a non-zero linear combination of the rows of $M_T$ is equal to zero, which contradicts the full-rankness of $M_T$, thus proving that $M_T(z)$ has to be full rank over $\mathbb{F}_q(z).$
\end{proof}
\subsection{Construction} 
\label{construction}
This subsection presents the main contribution of this work. We assume an $n$ dimensional network code ($n$ being the min-cut) on this network has implemented on the given network which is used to multicast information to a set of sinks. We describe a construction of an input convolutional code for the given acyclic, unit-delay, memory-free network which can correct network errors with patterns in a given error pattern set, as long as they are separated by certain number of network uses. 

Let $M_T(z)=AF_{T}(z)$ be the $n\times n$ network transfer matrix from the source to any particular sink $T\in {\cal T}$. Let $\Phi$ be the error pattern set given. We then define the \textit{processing matrix at sink T}, $P_T(z)$, to be a polynomial matrix as
\[
P_T(z)=p_{_T}(z)M_T^{-1}(z)
\]
where $p_{_T}(z) \in \mathbb{F}_q[z]$ is some \textit{processing function} chosen such that $P_T(z)$ is a polynomial matrix. Now, we have the construction of a CNECC for the given network as follows.
\begin{enumerate}
\item  We first compute the set of all error vectors having their error pattern in $\Phi$ that is defined as follows
\[
{\cal W}_{\Phi}=\bigcup_{\rho \in \Phi}\left\{\boldsymbol{w}=(w_1,w_2,...,w_{|{\cal E}|}) \in \mathbb{F}_{q}^{|{\cal E}|}~|~\boldsymbol{w} \in \rho\right\}.
\]
\item Let 
\begin{equation}
\label{eqn20}
{\cal W}_{T}:= \left\{\boldsymbol{w}F_T(z)~|~\boldsymbol{w}\in{\cal W}_{\Phi}\right\}
\end{equation}
be computed for each sink $T$. This is the set of $n$-tuples (with elements from $\mathbb{F}_q[z]$) at the sink $T$ due to errors in the given error patterns $\rho \in \Phi$. 
\item Let the set ${\cal W}_s \subset \mathbb{F}_{q}^{n}[z]$ 
\begin{equation}
\label{eqn5}
{\cal W}_s:=\bigcup_{T\in{\cal T}} \left\{\boldsymbol{w}_{_T}(z)P_T(z)~|~\boldsymbol{w}_{_T}(z)\in{\cal W}_{T}\right\} 
\end{equation}
be computed.  
\item Let 
\[
t_s = \max_{\boldsymbol{w}_s(z) \in {\cal W}_s}w_H\left(\boldsymbol{w}_s(z)\right).
\] where $w_H$ indicates the Hamming weight over $\mathbb{F}_q.$
\item Choose an input convolutional code ${\cal C}_s$ with free distance at least $2t_s+1$ as the CNECC for the given network. 
\end{enumerate}
\subsection{Decoding}
\label{decoding}

Before we discuss the decoding of CNECCs designed according to Subsection \ref{construction}, we state some of the results from \cite{PrR} related to the bounded distance decoding of convolutional codes in this section.


Let $\cal C$ be a rate $b/c$ convolutional code with a  generator matrix $G(z).$ Then, corresponding to the information sequence $\boldsymbol{u}_0,\boldsymbol{u}_1,.. (\boldsymbol{u}_i \in \mathbb{F}_q^b)$ and the codeword sequence $\boldsymbol{v}_0,\boldsymbol{v}_1,... (\boldsymbol{v}_i \in \mathbb{F}_q^c)$, we can associate an encoder state sequence $\boldsymbol{\sigma}_0,\boldsymbol{\sigma}_1,. . $, where 	$\boldsymbol{\sigma}_t$ indicates the content of the delay elements in the encoder at a time $t.$ We define the set of $j$ output symbols as
\[
\boldsymbol{v}_{[0,j)}:=\left[\boldsymbol{v}_0,\boldsymbol{v}_1,. . . ,\boldsymbol{v}_{j-1}\right]
\]

The parameter $T_{d_{free}}({\cal C})$ \cite{PrR} is defined as follows.
\[
T_{d_{free}}({\cal C}):=\max_{\boldsymbol{v}_{[0,j)} \in S_{d_{free}}}j+1
\]where $S_{d_{free}}$ \cite{PrR} is defined as the set of all possible truncated codeword sequences $\boldsymbol{v}_{[0,j)}$ of weight less than $d_{free}({\cal C})$ that start in the zero state is defined as follows
\begin{equation*}
\label{sdfree}
S_{d_{free}}:=\left\{\boldsymbol{v}_{[0,j)} \mid w_H\left(\boldsymbol{v}_{[0,j)}\right) < d_{free}({\cal C}),  \boldsymbol{\sigma}_0=\boldsymbol{0},\forall~j>0 \right\}
\end{equation*}
where $w_H$ indicates the Hamming weight over $\mathbb{F}_q.$


Then, we have the following proposition.
\begin{proposition}[\cite{PrR} ]
\label{minweighttime}
The minimum Hamming weight trellis decoding algorithm can correct all error sequences which have the property that the Hamming weight of the error sequence in any consecutive $T_{d_{free}}({\cal C})$ segments (a segment being the set of $c$ code symbols generated for every $c$ information symbols) is utmost $\left\lfloor \frac{d_{free}({\cal C})-1}{2} \right\rfloor$.
\end{proposition}

Now, we discuss the decoding of CNECCs for unit-delay memory-free networks. Let $G_I(z)$ be the $k \times n$ generator matrix of the input convolutional code, ${\cal C}_s$, obtained from the given construction. Let $G_{O,{T}}(z) = G_I(z)M_{T}(z)$ be the generator matrix of the output convolutional code, ${\cal C}_{T}$, at sink $T \in {\cal T}$, with $M_T(z)$ being its network transfer matrix. 

For each sink $T \in \cal T$, let
\[
t_T = \max_{\boldsymbol{w}_{_T}(z) \in {\cal W}_T} w_H(\boldsymbol{w}_{_T}(z)).
\]
Let $m_T$ be the largest integer such that 
\begin{equation}
\label{eqn11}
d_{free}({{\cal C}_T}) \geq 2m_Tt_T+1. 
\end{equation}
Clearly, $m_T \geq 0.$ Each sink can choose decoding on the trellis of the input or its output convolutional code based on the characteristics of the output convolutional code as follows

\textit{Case-A:} This is applicable in the event of all of the following conditions being satisfied. 
\begin{enumerate}
\renewcommand{\labelenumi}{\roman{enumi}.)} 
\item 
\begin{equation}
\label{decodAcond1} 
m_T \geq 1
\end{equation} 
\item
\begin{equation}
\label{decodAcond2} 
T_{d_{free}}({\cal C}_T) \leq m_TT_{d_{free}}({\cal C}_s).
\end{equation}
\item
The output convolutional code generator matrix $G_{O,T}(z)$ is non-catastrophic.
\vspace{-0.5cm}
\begin{equation}
\label{decodAcond3}
\end{equation}
\end{enumerate}
In this case, the sink $T$ performs minimum distance decoding directly on the trellis of the output convolutional code, ${\cal C}_{T}$.

\textit{Case-B:} This is applicable if at least one of the $3$ conditions of Case-A is not satisfied, i.e, if either of the following conditions hold
\begin{enumerate}
\renewcommand{\labelenumi}{\roman{enumi}.)} 
\item  
$
~~~m_T = 0
$
\item
$
~~~m_T \geq 1 ~~~and~~~~ T_{d_{free}}({\cal C}_T) > m_TT_{d_{free}}({\cal C}_s).
$
\item
The output convolutional code generator matrix $G_{O,T}(z)$ is catastrophic. 
\end{enumerate}
This method involves processing (matrix multiplication using $P_T(z)$) at the sink $T.$ We have the following formulation at the sink $T$. Let 
\begin{eqnarray*} 
\left[v_{1}'(z) ~~ v_{2}'(z) ~~ ... ~~ v_{n}'(z)\right] =  \left[v_{1}(z)~~v_{2}(z)~~ ... ~~v_{n}(z)\right]~~~~~~\\  ~~~~~~~~~~~~~~~~~~~~~~~+\left[w_{1}(z)~~w_{2}(z)~~...~~w_{n}(z)\right] 
\end{eqnarray*}
represent the output sequences at sink $T$, where 
\begin{eqnarray*}
\left[v_{1}(z) ~~ v_{2}(z) ~~ ... ~~ v_{n}(z)\right] ~~~~~~~~~~~~~~~~~~~ \\
= \boldsymbol{u}(z)G_{O,T}(z) = \boldsymbol{u}(z)G_{I}(z)M_T(z)
\end{eqnarray*}
$\boldsymbol{u}(z)$ being the $k$ length vector of input sequences, and 
\[
\left[w_{1}(z) ~~ w_{2}(z) ~~ ... ~~w_{n}(z)\right]
\]
represent the corresponding error sequences. Now, the output sequences are multiplied with the processing of the network transfer matrix $P_T(z)$, so that decoding can be done on the trellis of the input convolutional code. Hence, we have 
\begin{eqnarray} 
\nonumber
\left[v_{1}''(z) ~~ v_{2}''(z) ~~ ... ~~ v_{n}''(z) \right] ~~~~~~~~~~~~~~~~~~~~~~~~~~~~~~~~~~~~~~~~~~~~~~~~~\\
\nonumber
= \left[v_{1}'(z) ~~ v_{2}'(z) ~~ ... ~~ v_{n}'(z) \right]P_T(z)~~~~~~~~~~~~~~~~~~~~~~~~~~~~~~~~~ \\
\nonumber
= \boldsymbol{u}(z)p_{_T}(z)G_{I}(z) + \left[w_{1}(z) ~~ w_{2}(z) ~~ ... ~~w_{n}(z) \right]P_T(z)~~~~~~~~\\
\label{eq:3}
=\boldsymbol{u}(z)p_{_T}(z)G_{I}(z) + \left[w_{1}'(z) ~~ w_{2}'(z) ~~ ... ~~w_{n}'(z) \right] ~~~~~~~~~~~~~~~
\end{eqnarray}
where  $\boldsymbol{w'}(z) = \left[w_{1}'(z) ~~ w_{2}'(z) ~~ ... ~~w_{n}'(z) \right]$ now indicate the set of modified error sequences that are to be corrected. Now the sink $T$ decodes to the minimum distance path on the trellis of the code generated by $p_{_T}(z)G_{I}(z)$, which is the input convolutional code as $G_{I}(z)$ and $p_{_T}(z)G_{I}(z)$ are equivalent generator matrices. 
\begin{remark}
In \cite{PrR}, the approach to the construction of a CNECC for an instantaneous network was the same as in here. However, the set ${\cal W}_s$ was defined in \cite{PrR} as
\begin{eqnarray}
\nonumber
{\cal W}_s:=\bigcup_{T\in{\cal T}} \left\{\boldsymbol{w}_{_T}M_T^{-1}~|~\boldsymbol{w}_{_T}\in{\cal W}_{T}\right\}~~~ \\
\label{eqn12}
~~~~~~~~~~~~~~~~~~~~~~~= \bigcup_{T\in{\cal T},\rho \in \Phi} \left\{\boldsymbol{w}F_TM_T^{-1}~|~\boldsymbol{w}\in \rho\right\}
\end{eqnarray}
where the network transfer matrix $M_T$ and $F_T$ correspond to a sink $T$ in the instantaneous network.	

In this paper, the definition for ${\cal W}_s$ is as in (\ref{eqn5}) and involves the processing matrix $P_T(z)$ instead of the inverse of the network transfer matrix. The processing function $p_{_T}(z)$ for a sink $T$ is introduced because of the fact that the matrix $M_T^{-1}(z)$ might not be realizable and also for easily obtaining the Hamming weight of the \textit{error vector reflections} $\left(\boldsymbol{w}_s(z) \in {\cal W}_s\right)$ by removing rational functions in $M_T^{-1}(z).$ 

The degree of the processing function $p_{_T}(z)$ directly influences the memory requirements at the sinks and therefore should be kept as minimal as possible. Therefore, with 
\[
M_T^{-1}(z)=\frac{M(z)}{Det\left(M_T(z)\right)}
\]
where the $n \times n$ matrix $M(z)$ is the adjoint of $M_T(z)$, ideally we may choose $p_{_T}(z)$ as follows.
\begin{equation}
\label{eqn6}
p_{_T}(z) = \frac{Det\left(M_T(z)\right)}{g(z)}
\end{equation}
where $g(z)=gcd\left(m_{i,j}(z),~\forall~ 1 \leq i,j \leq n\right)$, $m_{i,j}(z)$ being the $(i,j)^{th}$ element of $M(z).$
\end{remark}
\subsection{Error correcting capability}
\label{capability}
In this subsection we prove a main result of the paper given by Theorem \ref{maintheorem} which characterizes the error correcting capability of the code obtained via the construction of Subsection \ref{construction}. We recall the following observation that in every network use, $n$ encoded symbols which is equal to the number of symbols corresponding to one segment of the trellis, are to be multicast to the sinks.
\begin{theorem}
\label{maintheorem}
The code ${\cal C}_s$ resulting from the construction of Subsection \ref{construction} can correct all network errors that have their pattern as some $\rho \in \Phi$ as long as any two consecutive network errors are separated by $T_{d_{free}}({\cal C}_s)$ network uses.
\end{theorem}
\begin{IEEEproof}
We first prove the theorem in the event of Case-A of the decoding. Suppose the network errors are such that consecutive network errors are separated by $T_{d_{free}}({\cal C}_s)$ network uses. Then the vector of error sequences at sink $T$,  $\boldsymbol{w}_{_T}(z)$, is such that in every $T_{d_{free}}({\cal C}_s)$ segments, the error sequence has utmost $t_T$ Hamming weight (over $\mathbb{F}_q$). Therefore in $m_TT_{d_{free}}({\cal C}_s)$ segments, the Hamming weight of the error sequence would be utmost $m_Tt_T.$

Then the given condition (\ref{decodAcond2}) would imply that in every $T_{d_{free}}({\cal C}_T)$ segments of the output trellis, the error sequences have Hamming weight utmost $m_Tt_T.$ Condition (\ref{decodAcond1}) together with (\ref{eqn11}) and Proposition \ref{minweighttime} implies that these error sequences are correctable. This proves the given claim that errors with their error pattern in $\Phi$  will be corrected as long as no two consecutive error events occur within $T_{d_{free}}({\cal C}_s)$ network uses. 

In fact, condition (\ref{decodAcond1}) and (\ref{eqn11}) implies that network errors with pattern in $\Phi$ will be corrected at sink $T$, as long as consecutive error events are separated by $T_{d_{free}}({\cal C}_{T})$.

Now we consider Case B of the decoding. Suppose that the set of error sequences in the formulation given, $\boldsymbol{w'}(z)$, is due to network errors that have their pattern as some $\rho \in \Phi$, such that any two consecutive such network errors are separated by  at least $T_{d_{free}}({\cal C}_s)$ network uses. 

Therefore, along with step $4$ of the construction, we have that the maximum Hamming weight of the error sequence $\boldsymbol{w'}(z)$ in any consecutive $T_{d_{free}}({\cal C}_s)$ segments (network uses) would be utmost $t_s$. Because of the free distance of the code chosen and along with Proposition \ref{minweighttime}, we have that such errors will get corrected when decoding on the trellis of the input convolutional code.
\end{IEEEproof}
\subsection{Bounds on the field size and $T_{d_{free}}({\cal C}_s)$}
\label{sec4e}
\subsubsection{Bound on field size}
Towards obtaining a bound on the sufficient field size for the construction of a CNECC meeting our free distance requirement, we first prove the following lemmas.
\begin{lemma}
\label{tdelay}
Given an acyclic, unit-delay, memory-free network ${\cal G}({\cal V},{\cal E})$ with a given error pattern set $\Phi$, let $T_{delay}-1	$ be the maximum degree of any polynomial in the $F(z)$ matrix. Let $w_H$ indicate the Hamming weight over $\mathbb{F}_q.$ If $r$ is the maximum number of non-zero coefficients of the polynomials $p_{_T}(z)$ corresponding to all sinks in $\cal T$, i.e 
\[
r=\max_{T \in {\cal T}} w_H\left(p_{_T}(z)\right),
\]
then
\[
\max_{\boldsymbol{w}_s(z) \in {\cal W}_s}w_H\left(\boldsymbol{w}_s(z)\right) \leq rn\left[\left(n+1\right)\left(T_{delay}-1\right)+1\right].
\]
where ${\cal W}_s$ is as in (\ref{eqn5}) in Subsection \ref{construction}.
\end{lemma}
\begin{proof}
Any element $\boldsymbol{w}_s(z)\in{\cal W}_s$ indicates the $n$ length sequences that would result in an output vector $\boldsymbol{w}_{_T}(z)$ at some sink $T$ as a result of an error vector $\boldsymbol{w}$ in the network at time $0$, i.e 
\[
\boldsymbol{w}_s(z)=\boldsymbol{w}F_T(z)p_{_T}(z)M_T^{-1}(z)=\boldsymbol{w}_{_T}(z)p_{_T}(z)M_T^{-1}(z)
\]

Because of the fact that any polynomial in $F(z)$ has degree utmost $T_{delay}-1$, any error vector $\boldsymbol{w}$ at time $0$ can result in non-zero symbols (over $\mathbb{F}_q^n$) in $\boldsymbol{w}_{_T}(z)$ at any sink $T$ from the $0^{th}$ time instant only upto utmost $T_{delay}-1$ time instants.
\[
\boldsymbol{w}_{_T}(z)=\left(\sum_{i=0}^{T_{delay}-1}\boldsymbol{w}_{_{T,i}}z^i \right).
\]
where $\boldsymbol{w}_{_{T,i}} \in \mathbb{F}_q^n.$ 

The numerator polynomial of any element $a(z) \in \mathbb{F}_q(z)$ of the matrix $M_T^{-1}(z)$ has degree utmost $n\left(T_{delay}-1\right)$. Therefore, considering the polynomial processing matrix $P_T(z)=p_{_T}(z)M_T^{-1}(z)$, we note that any element from $P_T(z)$ has utmost $r\left[n\left(T_{delay}-1\right)+1\right]$ non-zero components (over $\mathbb{F}_q$), the worst case being $r$ non-overlapping `blocks' of $n\left(T_{delay}-1\right)+1$ non-zero components each. 

Therefore the first non-zero symbol of $\boldsymbol{w}_{_T}(z)$ (over $\mathbb{F}_q^n$) at some time instant can result in utmost $r\left[n\left(T_{delay}-1\right)+1\right]$ non-zero symbols in $\boldsymbol{w}_s(z)$ (over $\mathbb{F}_q^n$). Henceforth, every consecutive non-zero symbol (over $\mathbb{F}_q^n$) of  $\boldsymbol{w}_{_T}(z)$ will result in utmost additional $r$ $\mathbb{F}_q^n$ symbols in $\boldsymbol{w}_s(z).$ Therefore any $\boldsymbol{w}_s(z) \in {\cal W}_s$ is of the form  
\[
\boldsymbol{w}_s(z)=\left(\sum_{i=0}^{r\left[\left(n+1\right)\left(T_{delay}-1\right)+1\right]}\boldsymbol{w}_{s,i}z^i \right)
\]
where $\boldsymbol{w}_{s,i}\in\mathbb{F}_q^n.$ Therefore the Hamming weight (over $\mathbb{F}_q$) of any $\boldsymbol{w}_s(z) \in {\cal W}_s$ is utmost $rn\left[\left(n+1\right)\left(T_{delay}-1\right)+1\right]$, thus proving the lemma. 
\end{proof}
Our bound on the field size requirement of CNECCs for unit-delay networks is based on the bound on field size for the construction of Maximum Distance Separable (MDS) convolutional codes \cite{RoS}, a primer on which can be found in Appendix \ref{app2}.
\begin{lemma}
\label{deltabound}
A $(n,k)$ MDS convolutional code $\cal C$ (over some field $\mathbb{F}_q$) with degree $\delta = \left\lceil \left(2t-1\right)k/n \right\rceil$ can correct any error sequence which has the property that the Hamming weight(over $\mathbb{F}_q$) of the error sequence in any consecutive $T_{d_{free}}({\cal C})$ segments is utmost $t.$ 
\end{lemma}
\begin{IEEEproof}
Because the generalized Singleton bound is satisfied with equality by the MDS convolutional code, we have
\[
d_{free}({\cal C})=(n-k)(\lfloor \delta / k \rfloor + 1) + \delta + 1.
\]
Substituting $\left\lceil \left(2t-1\right)k/n \right\rceil$ for $\delta$, we have
\begin{eqnarray*}
d_{free}({\cal C})~~~~~~~~~~~~~~~~~~~~~~~~~~~~~~~~~~~~~~~~~~~~~~~~~~~~~~~~~~~~~~~~~~~~\\
=(n-k)\left(\frac{\left\lceil \left(2t-1\right)k/n \right\rceil}{k} + 1\right) + \left\lceil \left(2t-1\right)k/n \right\rceil + 1 ~~~~~~~~~~~\\
d_{free}({\cal C}) \geq (n-k)\left(\frac{\left(2t-1\right)}{n} + 1 \right) + \frac{\left(2t-1\right)k}{n} + 1~~~~~~~~~~~~~ \\
\Longrightarrow d_{free}({\cal C}) \geq 2t+1.~~~~~~~~~~~~~~~~~~~~~~~~~~~~~~~
\end{eqnarray*}
Thus the free distance of the code $\cal C$ is at least $2t+1$, and therefore by Proposition \ref{minweighttime}, such a code can correct all error sequences which have the property that in any consecutive $T_{d_{free}}({\cal C})$ segments, the Hamming weight (over $\mathbb{F}_q$) of the error sequence is utmost $t.$ 
\end{IEEEproof}

For an MDS convolutional code being chosen as the input convolutional code (CNECC), we therefore have the following corollary
\begin{corollary}
\label{deltaboundnec}
Let ${\cal G}({\cal V},{\cal E})$ be an acyclic, unit-delay, memory-free network with a network code over a sufficiently large field $\mathbb{F}_q$ and $\Phi$ be an error pattern set, the errors corresponding to which are to be corrected. An $(n,k)$ input MDS convolutional code ${\cal C}_s$ over $\mathbb{F}_q$ with degree $\delta=2rk\left[\left(n+1\right)\left(T_{delay}-1\right)+1\right]$ can be used to correct all network-errors with their error pattern in $\Phi$ provided that consecutive network-errors are separated by at least $T_{d_{free}}({\cal C}_s)$ network uses, where $r$ and $T_{delay}$ are as in Lemma \ref{tdelay}. 
\end{corollary}
\begin{IEEEproof}
From Lemma \ref{tdelay}, we have that in the construction of Subsection \ref{construction}, the maximum Hamming weight $t_s$ of any element in the set ${\cal W}_s$ is utmost  $rn\left[\left(n+1\right)\left(T_{delay}-1\right)+1\right].$ For an input MDS convolutional code ${\cal C}_s$ to be capable of correcting such errors with Hamming weight utmost $rn\left[\left(n+1\right)\left(T_{delay}-1\right)+1\right]$, according to Lemma \ref{deltabound}, a degree $\delta = 2rk\left[\left(n+1\right)\left(T_{delay}-1\right)+1\right]$ would suffice.
\end{IEEEproof}

The following theorem gives a sufficient field size for the required network error correcting $(n,k)$ input convolutional code ${\cal C}_s$ to be constructed with the required free distance condition ($d_{free}({\cal C}_s) \geq 2t_s+1$).
\begin{theorem}
\label{fieldsizebound}
The code ${\cal C}_s$ can be constructed and used to multicast $k$ symbols to the set of sinks ${\cal T}$ along with the required error correction in the given acyclic, unit-delay, memory-free network with min-cut $n$ ($n>k$), if the field size $q$ is such that 
\begin{eqnarray*}
n|(q-1)~~~~~~~~~~~~~~~~~~~~~~~~~~~~\\
\text{and}~~~~~~~~~~~~~~~~~~~~~~~~~~~~~~\\
~~q > max\left\{|{\cal T}|,\frac{2rn^2\left[\left(n+1\right)\left(T_{delay}-1\right)+1\right]}{n-k}+2\right\}.
\end{eqnarray*}
\end{theorem}
\begin{IEEEproof}
From the sufficient condition for the existence of a linear multicast network code for a single source network with a set of sinks $\cal T$, we have 
\[
q>|{\cal T}|.
\]

Now we prove the other conditions. From the construction in \cite{RLS}, we know that a $(n,k,\delta)$ MDS convolutional code can be constructed over $\mathbb{F}_q$ if
\[
n|(q-1)~~~~~~~~  and~~~~~~~  q >\frac{{\delta}n^2}{k\left(n-k\right)}+2.
\]

Thus, with $\delta = 2rk\left[\left(n+1\right)\left(T_{delay}-1\right)+1\right]$ as in Corollary \ref{deltaboundnec}, an input MDS convolutional code ${\cal C}_s$ can be constructed over $\mathbb{F}_q$ if
\[
n|(q-1)~~~~~~ and~~~~~~q > \frac{2rn^2\left[\left(n+1\right)\left(T_{delay}-1\right)+1\right]}{n-k}+2.
\]
Such an MDS convolutional code the requirements in the construction $\left(d_{free}({\cal C}_s) \geq 2rn\left[\left(n+1\right)\left(T_{delay}-1\right)+1\right]+1 \geq 2t_s+1\right)$, and hence the theorem is proved. 
\end{IEEEproof}
\subsubsection{Bound on $T_{d_{free}}({\cal C}_s)$}
Towards obtaining a bound on $T_{d_{free}}({\cal C}_s)$, we first restate the following bound proved in \cite{PrR}. 
\begin{proposition}
Let $\cal C$ be a $(c,b,\delta)$ convolutional code. Then
\begin{equation}
\label{eqn19}
T_{d_{free}}({\cal C}) \leq \left(d_{free}\left({\cal C}\right)-1\right)\delta+1.
\end{equation}
\end{proposition}

Thus, for a network error correcting MDS convolutional code ${\cal C}_s$ for the unit-delay network, we have the following bound on $T_{d_{free}}({\cal C}_s)$.
\begin{corollary}
Let the CNECC ${\cal C}_s$ be a $(n,k,\delta = 2rk\left[\left(n+1\right)\left(T_{delay}-1\right)+1\right])$ MDS convolutional code, where $r$ and $T_{delay}$ are as in  Lemma \ref{tdelay}. Then
\begin{eqnarray*}
T_{d_{free}}({\cal C}_s) \leq 4r^2nk\left[\left(n+1\right)\left(T_{delay}-1\right)+1\right]^2~~~~~~~~~~~~~~~~~ \\ 
~~~~~~~~~~~~~~~~+ 2rk\left(n-k\right)\left[\left(n+1\right)\left(T_{delay}-1\right)+1\right] + 1.
\end{eqnarray*}
\end{corollary}
\begin{IEEEproof}
For MDS convolutional codes, we have 
\[
d_{free}({\cal C}) = (n-k)(\lfloor\delta/k\rfloor + 1) + \delta + 1
\]
With $\delta = 2rk\left[\left(n+1\right)\left(T_{delay}-1\right)+1\right]$, we have 
\begin{eqnarray*}
d_{free}({\cal C}_s) = (n-k)\left\{2r\left[\left(n+1\right)\left(T_{delay}-1\right)+1\right] + 1\right\}\\ ~~~~~~~~~~~~~~~~~~+ 2rk\left[\left(n+1\right)\left(T_{delay}-1\right)+1\right] + 1 \\
d_{free}({\cal C}_s) = 2rn\left[\left(n+1\right)\left(T_{delay}-1\right)+1\right] + n-k+1~
\end{eqnarray*}
Substituting this value of $d_{free}({\cal C}_s)$ and $\delta$ in (\ref{eqn19}), we have proved that 
\begin{eqnarray*}
T_{d_{free}}({\cal C}_s) \leq 4r^2nk\left[\left(n+1\right)\left(T_{delay}-1\right)+1\right]^2~~~~~~~~~~~~~~~~~ \\ 
~~~~~~~~~~~~~~~~+ 2rk\left(n-k\right)\left[\left(n+1\right)\left(T_{delay}-1\right)+1\right] + 1.
\end{eqnarray*}
\end{IEEEproof}
\section{Illustrative examples}
\label{sec5}
\subsection{Code construction for a modified Butterfly network:}
\label{subsec5a}
Let us consider the modified butterfly network as shown in Fig. \ref{fig:butterflydelay}, with one of the edges at the bottleneck node (of the original unmodified butterfly network) having twice the delay as any other edge, thus forcing an inter-generation linear combination at the bottleneck node. The local kernels at the node defining the network code are the same as in that of the instantaneous butterfly case. We assume the network code to be over $\mathbb{F}_2$ and we design a convolutional code over $\mathbb{F}_2$ that will correct all single edge errors in the network, i.e, all network error vectors of Hamming weight utmost $1.$

\begin{figure}[htbp]
\centering
\includegraphics[totalheight=3.4in,width=3.2in]{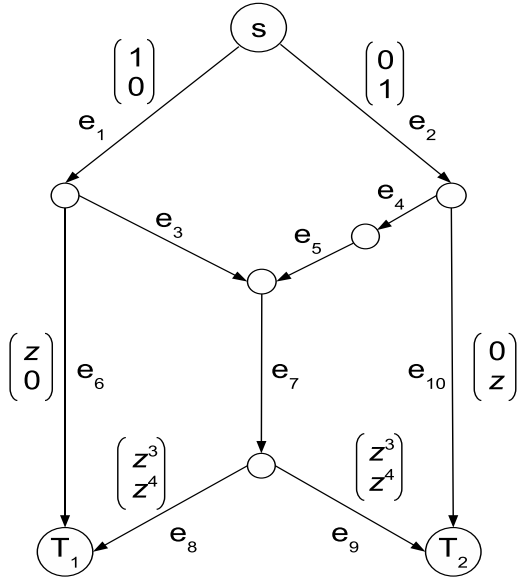}
\caption{Modified butterfly network with global kernels}	
\label{fig:butterflydelay}	
\end{figure}

For this network, the matrix $A$ is a $2 \times 10$ matrix having a $2 \times 2$ identity submatrix at the columns corresponding to edges $e_1$ and $e_2$, and having zeros everywhere else. We assume $B^{T_1}$ and $B^{T_2}$ are $10 \times 2$ matrices such that they have a $2 \times 2$ identity submatrix at rows $e_6, e_8$ and $e_9, e_{10}$ respectively. With the given network code, we thus have the network transfer matrices at sink $T_1$ and $T_2$ as follows
\[
M_{T_1}(z) = \left[ \begin{array}{cc}
z & z^3 \\
0  & z^4 \end{array} \right]=AF_{T_1}(z)
\]
where 
\[
F_{T_1}(z) = \left[ \begin{array}{cccccccccc}
z & 0 & 0 & 0 & 0 & 1 & 0 & 0 & 0 & 0\\
z^3 & z^4 & z^2 & z^3 & z^2 & 0 & z & 1 & 0 & 0
\end{array} \right]^T
\]
and 
\[
M_{T_2}(z) = \left[ \begin{array}{cc}
z^3 & 0 \\
z^4  & z \end{array} \right]=AF_{T_2}(z)
\]
where 
\[
F_{T_2}(z) = \left[ \begin{array}{ccccccccccc}
z^3 & z^4 & z^2 & z^3 & z^2 & 0 & z & 0 & 1 & 0\\
0 & z & 0 & 0 & 0 & 0 & 0 & 0 & 0 & 1
\end{array} \right]^T.
\]

For single edge errors, we have the error pattern set to be
\[
\Phi=\left\{\left\{e_i\right\}:i=1,2,...,9,10\right\}.
\]
And thus the set ${\cal W}_{\Phi}$ is the set of all vectors $\mathbb{F}_2$ that have Hamming weight utmost $1.$ The sets ${\cal W}_{T_1}$ and ${\cal W}_{T_2}$ as in (\ref{eqn13}) and (\ref{eqn14}) at the top of the next page.
\begin{figure*}
\begin{eqnarray}
\label{eqn13} {\cal W}_{T_1}= \left\{(0,0),(0,1),(1,0),(0,z),(0,z^2),(0,z^3),(0,z^4),(z,z^3)\right\} \\
\label{eqn14} {\cal W}_{T_2}= \left\{(0,0),(0,1),(1,0),(z,0),(z^2,0),(z^3,0),(0,z^4),(z^4,z)\right\}
\end{eqnarray}
\hrule
\end{figure*}
Now 
\[
M_{T_1}^{-1}(z) = \frac{1}{z^5}\left[ \begin{array}{cc}
z^4 & z^3 \\
0  & z \end{array} \right]
\]
and 
\[
M_{T_2}^{-1}(z) = \frac{1}{z^4}\left[ \begin{array}{cc}
z & 0 \\
z^4  & z^3 \end{array} \right].
\]
To obtain the processing matrices $P_{T_1}(z)$ and $P_{T_2}(z)$, let us choose the processing functions $p_{_{T_1}}(z) = z^4$ and $p_{_{T_2}}(z) = z^3.$ Then we have
\begin{equation}
\label{proc1}
P_{T_1}(z) = p_{_{T_1}}(z)M_{T_1}^{-1}(z)=\left[ \begin{array}{cc}
z^3 & z^2 \\
0  & 1 \end{array} \right]
\end{equation}
and 
\begin{equation}
\label{proc2}
P_{T_2}(z) = p_{_{T_2}}(z)M_{T_2}^{-1}(z)=\left[ \begin{array}{cc}
1 & 0 \\
z^3  & z^2 \end{array} \right].
\end{equation}
Therefore, ${\cal W}_s$ can be computed to be as in (\ref{eqn15}) at the top of the next page.
\begin{figure*}
\begin{equation}
\label{eqn15}
{\cal W}_s = \left\{ (0,0),(z^3,z^2), (0,1), (0,z),(0,z^2), (0,z^3),(0,z^4), (z,0), (z^2,0), (z^3,0) \right\}.
\end{equation}
\hrule
\end{figure*}
Thus we have $t_s=2$, which means that we need a convolutional code with free distance at least $5.$ Let the chosen input convolutional code ${\cal C}_s$ be generated by the generator matrix 
\[
G_I(z)=\left[1+z^2~~~1+z+z^2\right].
\]
This code has a free distance $d_{free}({\cal C}_s)=5$ and $T_{d_{free}}({\cal C}_s)=6.$ Therefore this code can be used to correct single edge errors in the butterfly network as long as consecutive errors are separated by $6$ network uses. With this code, the output convolutional code ${\cal C}_{T_1}$ at sink $T_1$ is generated by the matrix
\[
G_{O,T_1}(z)=\left[z + z^3~~z^3+z^4+z^6\right]
\]
Now ${\cal C}_{T_1}$ has $d_{free}({\cal C}_{T_1})=5$ and $T_{d_{free}}({\cal C}_{T_1})=9 > T_{d_{free}}({\cal C}_s)$. As condition (\ref{decodAcond2}) is not satisfied, Case-B applies and hence the sink $T_1$ has to use the processing matrix $P_{T_1}(z)$, and then decode on the trellis of the input convolutional code. Upon performing a similar analysis for sink $T_2$, we have Table \ref{tab1} as shown at the top of the next page.

\begin{table*}
\centering
\normalsize
\caption{Modified butterfly network with ${\cal C}_s[d_{free}({\cal C}_s)=5,T_{d_{free}}({\cal C}_s)=6]$} \begin{tabular}{|c|c|c|c|}\hline
\textbf{Sink} & \textbf{Output convolutional code generator matrix} $[G_{O,T_i}(z)]$& \textbf{$d_{free}({\cal C}_{T_i})$}, $T_{d_{free}}({\cal C}_{T_i})$& \textbf{Decoding on}\\
\hline
$T_1$ & $\left[z + z^3~~~z^3+z^4+z^6\right]$ & 5,9 & Input trellis \\
\hline
$T_2$ & $[z^3+z^4+z^6~~~z+z^2+z^3]$ & 6,12 & Input trellis\\
\hline
\end{tabular}
\label{tab1}
\newline \newline
\end{table*}
\subsection{$_4C_2$ combination network over ternary field}
\label{subsec5b}
We now give a code construction for double edge error correction in the $_4C_2$ combination network with a network code over $\mathbb{F}_3$, shown in Fig. \ref{fig:4C2delay}  with the given $2$ dimensional network code, the network transfer matrices and the processing matrices (upon choosing the processing functions $p_{T_i}(z) = p_{T}(z) = z~\forall~1 \leq i \leq 6$) corresponding to the $6$ sinks are indicated in Table \ref{tab2}.

\begin{figure}[htbp]
\centering
\includegraphics[totalheight=3.1in,width=3.5in]{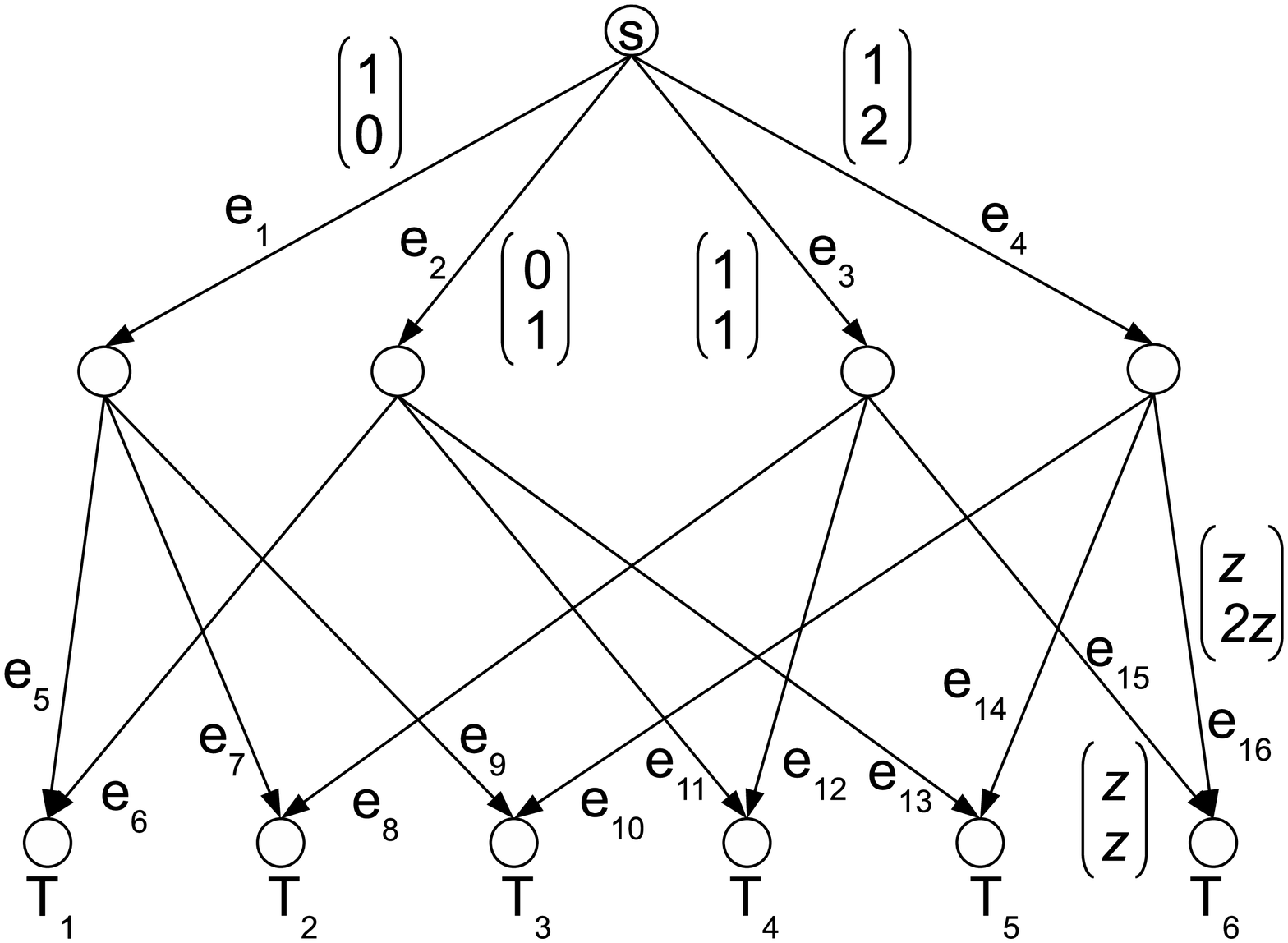}
\caption{$_4C_2$ unit-delay network}	
\label{fig:4C2delay}	
\end{figure}

The matrix $F_{T_1}(z)$ corresponding to sink $T_1$ is the $2 \times 16$ matrix as follows
\[
F_{T_1}(z) = \left[ \begin{array}{ccccccc}
z & 0 & 0 & 0 & 1 & 0 & 0......0\\
0 & z & 0 & 0 & 0 & 1 & 0......0
\end{array} \right]^T.
\]
For each sink, we have a similar $F_T(z)$ matrix with a $z$-scaled identity submatrix and an identity submatrix and zeros at all other entries.

For double edge error correction, the error pattern set $\Phi$ is
\[
\Phi~=~\left\{\left\{e_i,e_j\right\}:i,j=1,2,...,15,16 \text{ and }i \neq j\right\}.
\]
And therefore, we have the set ${\cal W}_{\Phi}$ as the set of all $16$ length tuples from $\mathbb{F}_3$ with Hamming weight utmost $2.$ The set ${\cal W}_{T_i}~\forall~i$ can be computed to be as shown in (\ref{eqn16}) at the top of the next page. Now, the set	 
\[
{\cal W}_{s,T_2} = \left\{\boldsymbol{w}_{_{T_2}}(z)P_{T_2}(z)~|~\boldsymbol{w}_{_{T_2}}(z)\in{\cal W}_{T_2}\right\}
\]
is computed to be as in (\ref{eqn17}), also shown at the top of the next page. 
\begin{figure*}
\begin{equation}
\label{eqn16}
{\cal W}_{T_i}=
\left\{ 
\begin{array}{ccccccc}
(z,0),& (0,z), & (1,0), & (0,1),& (2z,0), & (0,2z), & (2,0), \\
(0,2), & (z,z), & (z,2z), & (2z,z), & (2z,2z), & (z+1,0), & (z+2,0), \\
(2z+1,0), & (2z+2,0), & (z,1), & (z,2), & (2z,1), & (2z,2), & (1,z), \\ 
(1,2z), & (2,z), & (2,2z), & (0,z+1), & (0,z+2), & (0,2z+1), & (0,2z+2),\\
(1,1), & (1,2), & & (2,1), & &(2,2), & (0,0)
\end{array}
\right\}
\end{equation}
\hrule
\begin{equation}
\label{eqn17}
{\cal W}_{s,T_2}=
\left\{ 
\begin{array}{cccccc}
(z,2z),&(0,z), & (1,2), & (0,1), & (2z,z), & (0,2z), \\
(2,1), & (0,2), & (z,0), & (z,z), & (2z,2z), & (2z,0), \\ 
(z+1,2z+2), & (z+2,2z+1), & (2z+1,z+2), & (2z+2,z+1), & (z,2z+1), & (z,2z+2), \\
(2z,z+1), & (2z,z+2), & (1,z), & (1,2z), & (2,z), & (2,2z),\\
(0,z+1), & (0,z+2), & (0,2z+1), & (0,2z+2), & (1,0), & (1,1), \\
\text{ } & (2,2), && (2,0), && (0,0)
\end{array}
\right\}
\end{equation}
\hrule
\end{figure*}
\begin{table*}
\caption{$_4C_2$ combination network with ${\cal C}_s[d_{free}({\cal C}_s)=9,T_{d_{free}}({\cal C}_s)=14]$}
 \begin{tabular}{|c|c|c|c|c|c|}\hline
\textbf{Sink} & \textbf{Network transfer matrix} & \textbf{Processing matrix} & \textbf{Output convolutional code} & \textbf{$d_{free}({\cal C}_{T_i})$}, & \textbf{Decoding on}\\
 & & & \textbf{gen. matrix} $[G_{O,T_i}(z)]$ & $T_{d_{free}}({\cal C}_{T_i})$ & \\
\hline
$T_1$ & $M_{T_1}(z)=\left( \begin{array}{cc}z & 0 \\0  & z \end{array} \right)$ & $P_{T_1}(z)= \left( \begin{array}{cc}1 & 0 \\0  & 1 \end{array} \right)$ & $\left[z+z^3+z^5+z^6~~2z+z^2+2z^3+2z^5+z^6\right]$ & 5,9 & Output trellis \\
\hline
$T_2$ & $M_{T_2}(z)=\left( \begin{array}{cc}z & z \\0  & z \end{array} \right)$ & $P_{T_2}(z)= \left( \begin{array}{cc}1 & 2 \\0  & 1 \end{array} \right)$ & $[z+z^3+z^5+z^6~~z^2+2z^6]$ & 6,11 & Output trellis\\
\hline
$T_3$ & $M_{T_3}(z)=\left( \begin{array}{cc}z & z \\0  & 2z \end{array} \right)$ & $P_{T_3}(z)= \left( \begin{array}{cc}2 & 2 \\0  & 1 \end{array} \right)$ & $[z+z^3+z^5+z^6~~2z+2z^2+2z^3+2z^5]$ & 6,11 & Output trellis\\
\hline
$T_4$ & $M_{T_4}(z)=\left( \begin{array}{cc}0 & z \\z  & z \end{array} \right)$ & $P_{T_4}(z)= \left( \begin{array}{cc}1 & 2 \\2  & 0 \end{array} \right)$ & $[2z+z^2+2z^3+2z^5+z^6~~z^2+2z^6]$ & 7,12 & Output trellis\\
\hline
$T_5$ & $M_{T_5}(z)=\left( \begin{array}{cc}0 & z \\z  & 2z \end{array} \right)$ & $P_{T_5}(z)= \left( \begin{array}{cc}2 & 2 \\2  & 0 \end{array} \right)$ & $[2z+z^2+2z^3+2z^5+z^6~~2z+2z^2+2z^3+2z^5]$ & 9,14 & Output trellis \\
\hline
$T_6$ & $M_{T_6}(z)=\left( \begin{array}{cc}z & z \\z  & 2z \end{array} \right)$ & $P_{T_6}(z)= \left( \begin{array}{cc}2 & 2 \\2  & 1 \end{array} \right)$ & $[z^2+2z^6~~2z+2z^2+2z^3+2z^5]$ & 6,13 & Output trellis\\
\hline
\end{tabular}
\label{tab2}
\newline \newline
\end{table*}
Similarly the sets ${\cal W}_{s,T_i}~(\forall~1 \leq i \leq 6)$ and  
\[
{\cal W}_s=\bigcup_{T_i\in{\cal T}}{\cal W}_{s,T_i}
\]
are computed. It is seen that for this network, 
\[
t_s = \max_{\boldsymbol{w}_s(z) \in {\cal W}_s}w_H\left(\boldsymbol{w}_s(z)\right) = 4
\]
and
\[
t_{T_i} = \max_{\boldsymbol{w}_{_{T_i}}(z) \in {\cal W}_{T_i}}w_H\left(\boldsymbol{w}_{_{T_i}}(z)\right) = 2,~\forall~1 \leq i \leq 6.
\]
Therefore we need a convolutional code with free distance $9$ to correct such errors. Let this input convolutional code ${\cal C}_s$ over $\mathbb{F}_3$ be chosen as the code generated by 
\[
G_I(z)=\left[1+z^2+z^4+z^5~~2+z+2z^2+2z^4+z^5\right].
\] 
This code is found to have $d_{free}({\cal C}_s)=9$ with $T_{d_{free}}({\cal C}_s)=14.$ Thus it can correct all double edge network errors as long as consecutive network errors are separated by $14$ network uses. The output convolutional codes $T_{d_{free}}({\cal C}_{T_i})$, their free distance and $T_{d_{free}}({\cal C}_{T_i})$ are computed and tabulated in  Table \ref{tab2} at the top of the next page. For this example, all the sinks satisfy the conditions (\ref{decodAcond1}) and (\ref{decodAcond2}) for Case-A of the decoding and therefore decode on the trellises of the corresponding output convolutional codes.
\section{Comparison between CNECCs for instantaneous and unit-delay, memory-free networks}
\label{sec6}
In the following discussion, we compare the CNECCs for a given instantaneous network constructed in \cite{PrR} and the CNECCs of Subsection \ref{construction} for the corresponding unit-delay, memory-free network.

With the given acyclic graph ${\cal G}({\cal V},{\cal E})$, we will compare the maximum Hamming weight $t_s$ of any $n$-tuple, over $\mathbb{F}_q[z]$ ($\boldsymbol{w}_s(z) \in {\cal W}_s$, where ${\cal W}_s$ is as in (\ref{eqn5})) in the case of the unit-delay, memory-free network with the  graph ${\cal G}$ and over $\mathbb{F}_q$ ($\boldsymbol{w}_s \in {\cal W}_s$ where ${\cal W}_s$ is as in (\ref{eqn12})) in the case of instantaneous network with the graph ${\cal G}$. 

Consider some $\boldsymbol{w}_s(z) \in {\cal W}_s$ such that
\begin{eqnarray}
\nonumber
\boldsymbol{w}_s(z)=\boldsymbol{w}F_T(z)P_T(z)=\boldsymbol{w}p_{_T}(z)F_T(z)M_T^{-1}(z) \\
\label{eqn7} = \left[w_{s,1}(z),w_{s,2}(z),...,w_{s,n}(z)\right]
\end{eqnarray}
where $p_{_T}(z)$ and $P_T(z)$ indicate the processing function and matrix chosen according to (\ref{eqn6}) for some sink $T \in {\cal T}$, and $w_{s,i}(z) \in \mathbb{F}_q[z].$ We have $M_T(z)|_{z=1} = M_T$ and also $F_T(z)|_{z=1}=F_T$, the network transfer matrix and the $F_T$ matrix of the sink $T$ in the instantaneous network. Now, by (\ref{eqn7}), we have the $n$-length vector $\boldsymbol{w}_{s,inst}$ corresponding to the error vector $\boldsymbol{w}$ as 
\[
\boldsymbol{w}_{s,inst}=\boldsymbol{w}F_TM_T^{-1}=\frac{\boldsymbol{w}_s(z)|_{z=1}}{p_{_T}(z)|_{z=1}}
\]
where
\[
p_{_T}(z)|_{z=1} = \frac{Det\left(M_T(z)\right)|_{z=1}}{g(z)|_{z=1}}
\]
by (\ref{eqn6}). Now $Det\left(M_T(z)\right)|_{z=1} = Det\left(M_T\right) \neq 0$ since $M_T$ is full rank. Also, $g(z)|_{z=1} \neq 0$ for the same reason. Therefore, $p_{_T}(z)|_{z=1} \neq 0.$ Thus we have 
\begin{equation}
\label{eqn8}
w_H\left(\boldsymbol{w}_{s,inst}\right) \leq w_H\left(\boldsymbol{w}_s(z)\right).
\end{equation}
Therefore a CNECC for an instantaneous network may require a lesser free distance to correct networks errors matching one of the given set of patterns $\Phi$, while the CNECC for the corresponding unit-delay, memory-free network may require a larger free distance to provide the same error correction according to the construction of Subsection \ref{construction}. 

An example of this case is the code construction for double edge error correction for the $_4C_2$ combination instantaneous network in \cite{PrR} and for the $_4C_2$ unit-delay network in this paper in Subsection \ref{subsec5b}. It can be seen that while for the instantaneous network, the maximum Hamming weight of any $\boldsymbol{w}_s \in {\cal W}_s$ is $2$, the maximum Hamming weight of any $\boldsymbol{w}_s(z) \in {\cal W}_s$ in the unit-delay network is $4.$ Thus a code with free distance $5$ suffices for the instantaneous network, while the code for the unit-delay network has to have a free distance $9$ to ensure the required error correction as per the construction in Subsection \ref{construction}. 

It is in general not easy to obtain the general conditions under which equality will hold in (\ref{eqn8}), as both the topology and the network code of the network influence the Hamming weight of any element in ${\cal W}_s.$ For specific examples however, this can be checked. An example of this case is given in between the single edge-error correcting code construction for the butterfly network (over $\mathbb{F}_2$) for the instantaneous case in \cite{PrR} (the additional intermediate node, $head(e_4)=v_3=tail(e_5)$, does not matter for the instantaneous case), and for the unit-delay case in this paper in Subsection \ref{subsec5a}. In both the cases, we have $t_s=2$, which means that an input convolutional code with free distance $5$ is sufficient to correct all single edge network errors. However, as we see in Subsection \ref{subsec5a}, processing matrices with memory elements need to be used at the sinks for the unit-delay case, while the processing matrix in the instantaneous case is just the $M_T^{-1}$ matrix which does not require any memory elements to implement.

\section{Simulation results}
\label{sec7}
\subsection{A probabilistic error model}
We define a probabilistic error model for a unit delay network  ${\cal G}({\cal V},{\cal E})$ by defining the probabilities of any set of $i~(i\leq|{\cal E}|)$ edges of the network being in error at any given time instant as follows. Across time instants, we assume that the network errors are i.i.d. according to this distribution.
\begin{align}
\label{eq:1}
Prob.&(i~\text{network edges being in error}) = p^i \\
\label{eq:2}
Prob.&(\text{no edges are in error}) = q
\end{align}
where $1 < i \leq |{\cal E}|,$ and $p,q \leq 1$ are real numbers indicating the probability of any single edge error in the network and the probability of no edges in error respectively, such that $q + \sum_{i=1}^{|{\cal E}|}p^i = 1.$
\subsection{Simulations on the modified butterfly network}
With the probability model as in  (\ref{eq:1}) and (\ref{eq:2}) with $|{\cal E}|=10$ for the modified butterfly network as in Fig. \ref{fig:butterflydelay}, we simulate the performance of $3$ input convolutional codes implemented on this network with the sinks performing hard decision decoding on the trellis of the input convolutional code. In the following discussion we refer to sinks $T_1$ and $T_2$ of Fig. \ref{fig:butterflydelay} as Sink 1 and Sink 2. The $3$ input convolutional codes and the rationality behind choosing them are given as follows. 
\begin{itemize}
\item Code ${\cal C}_1$ is generated by the generator matrix 
\[
G_{I_1}(z)=\left[1+z~~~1\right],
\]
with $d_{free}({\cal C}_1)=3$ and $T_{d_{free}}({\cal C}_1)=2.$ This code is chosen only to illustrate the error correcting capability of codes with low values of $d_{free}({\cal C})$ and $T_{d_{free}}({\cal C}).$ 
\item Code ${\cal C}_2$ is generated by the generator matrix 
\[
G_{I_2}(z)=\left[1+z^2~~~1+z+z^2\right],
\]
with $d_{free}({\cal C}_2)=5$ and $T_{d_{free}}({\cal C}_2)=6.$ This code corrects all double edge errors in the instantaneous version (with all edge delays being zero) of Fig. \ref{fig:butterflydelay} as long as they are separated by $6$ network uses.
\item  Code ${\cal C}_3$ is generated by the generator matrix 
\[
G_{I_3}(z)=\left[1+z+z^4~~~1+z^2+z^3+z^4\right],
\]
with $d_{free}({\cal C}_3)=7$ and $T_{d_{free}}({\cal C}_3)=12.$ This code corrects all double edge errors in the unit-delay network given in Fig. \ref{fig:butterflydelay} as long as they are separated by $12$ network uses.
\end{itemize}

We note here that values of $T_{d_{free}}({\cal C})$ of the $3$ codes are directly proportional to their free distances, i.e, the code with greater free distance has higher $T_{d_{free}}({\cal C})$. Also we note that with each of these $3$ codes as the input convolutional codes, the output convolutional codes violate at least one of the conditions of `Case-A' of decoding, i.e, (\ref{decodAcond1}),(\ref{decodAcond2}), or (\ref{decodAcond3}). Therefore, hard decision Viterbi decoding is performed on the trellis of the input convolutional code.

Fig. \ref{fig:BERsink1} and Fig. \ref{fig:BERsink2} illustrate the BERs for different values for the parameter $p$ (the probability of a single edge error) of (\ref{eq:1}). Clearly the BER values fall with decreasing $p.$
\begin{figure*}
\centering
\includegraphics[totalheight=4.9in,width=7.1in]{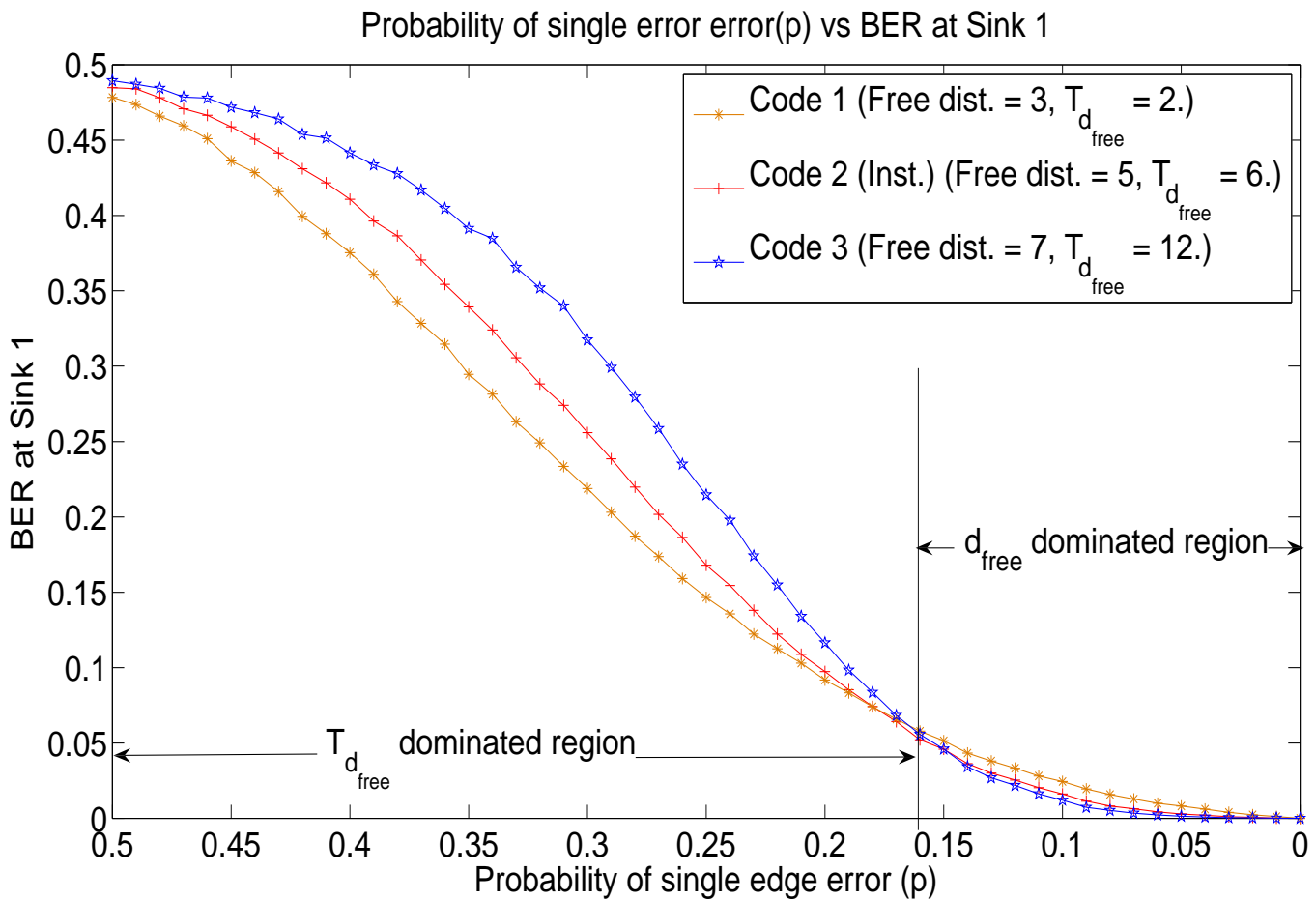}
\caption{BER at Sink 1}	
\label{fig:BERsink1}	
\hrule
\end{figure*}
\begin{figure*}
\centering
\includegraphics[totalheight=4.9in,width=7.1in]{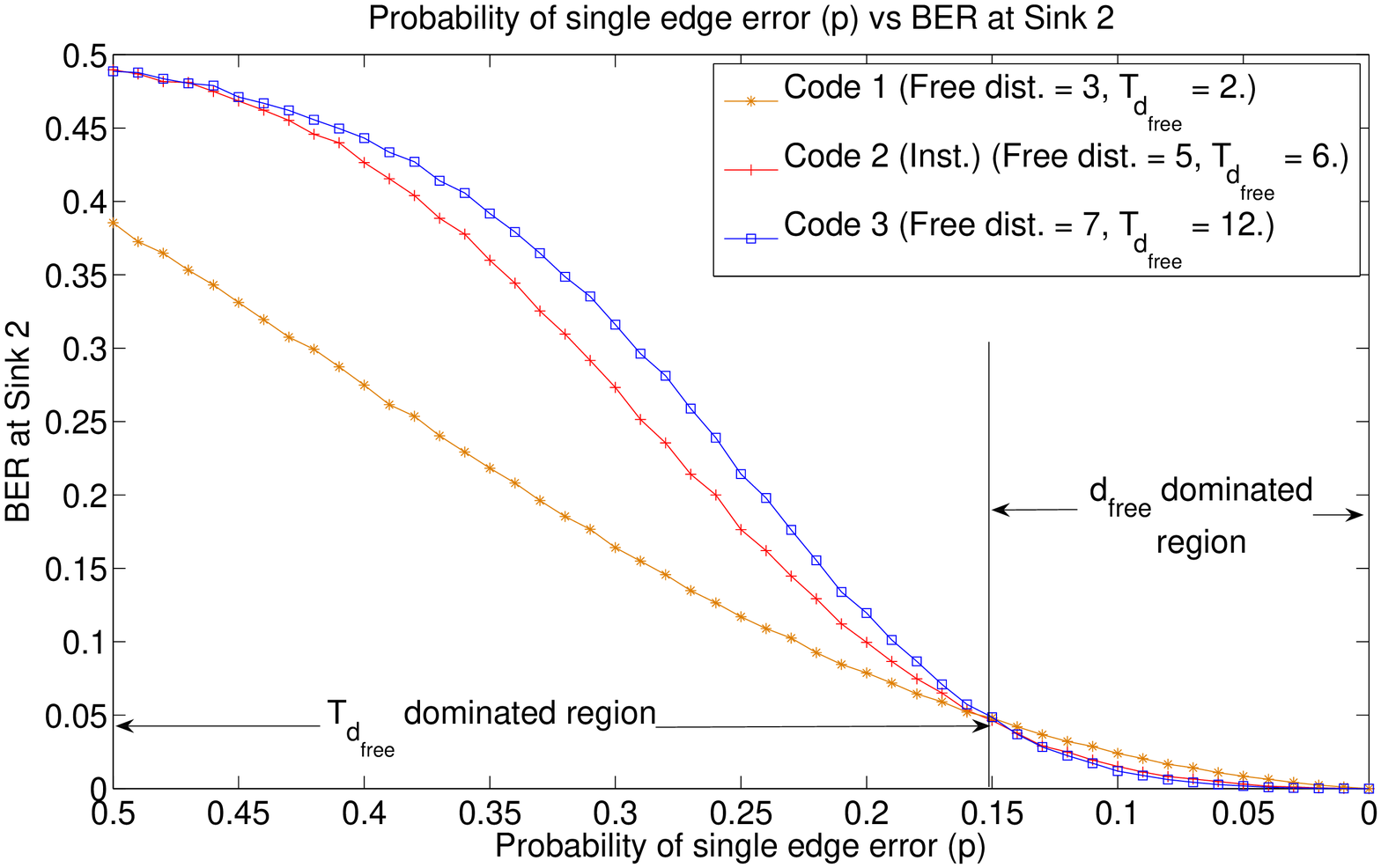}
\caption{BER at Sink 2}	
\label{fig:BERsink2}	
\hrule
\end{figure*}

It may be observed that between any two of the $3$ codes, say ${\cal C}_i$ and ${\cal C}_j$ ($i,j = 1,2,3$) there exist a particular value of $p=p_{i,j}$ where the BER performance corresponding to the two codes gets reversed, i.e, if code ${\cal C}_i$ has better BER performance than ${\cal C}_j$ for any $p > p_{i,j}$, then ${\cal C}_j$ performs better than ${\cal C}_i$ for any $p < p_{i,j}.$ Although such a cross-over value of $p$ exists for each pair of codes, we see that all $3$ codes have approximately the same crossover $p$ value in Fig. \ref{fig:BERsink1} ($p \approx 0.16$) and similarly in Fig. \ref{fig:BERsink2} ($p \approx 0.15$).

With respect to such crossover points between the two codes ${\cal C}_i$ and ${\cal C}_j$, we can divide the performance curve into two regions which we call as  `$T_{d_{free}}$ dominated region' ($p$ values being greater than the crossover $p$ value) and `$d_{free}$ dominated region' ($p$ values being lesser than the crossover $p$ value), indicating the parameter which controls the performance of the codes in each of those regions respectively. Again, because of the $3$ crossover points being approximately equal to one another in each of Fig. \ref{fig:BERsink1} and Fig. \ref{fig:BERsink2}, we divide the entire performance graph of all the $3$ codes into two regions. The following discussion gives an intuition into why the parameters $T_{d_{free}}({\cal C})$ and $d_{free}({\cal C})$ control the performance in the corresponding regions.  

\begin{itemize}
\item $d_{free}$ \textit{dominated region}: In the $d_{free}$ dominated region, codes with higher free distance perform better than those with less free distance. We recall from Proposition \ref{minweighttime} that both the Hamming weight of error events and the separation between any two consecutive error events are important to correct them. Because of the fact $p$ is low in the $d_{free}$ dominated region, the Hamming weight of the modified error sequences of (\ref{eq:3}) is less, and the error events that occur are also separated by sufficient number of network uses. Therefore the condition on the separation of error events according to Proposition \ref{minweighttime} is automatically satisfied even for large $T_{d_{free}}({\cal C})$ codes. Therefore codes which have more free distance (though having more $T_{d_{free}}({\cal C})$) correct more errors than codes with low free distance (though having less $T_{d_{free}}({\cal C})$). It is noted that in this region the code ${\cal C}_3$ (which was designed for correcting double edge errors on the unit-delay network) performs better than ${\cal C}_2$ (which was designed for correcting double edge errors on the instantaneous version of the network). 

\item $T_{d_{free}}$ \textit{dominated region}: In the $T_{d_{free}}$ dominated region, codes with lower  $T_{d_{free}}({\cal C})$ perform better than codes with higher $T_{d_{free}}({\cal C})$, even though their free distances might actually indicate otherwise. This is because of the fact that the error events related to the modified error sequences of (\ref{eq:3}) occur more frequently with lesser separation of network uses (as $p$ is higher). Therefore the codes with lower $T_{d_{free}}({\cal C})$ are able to correct more errors (even though the errors themselves must accumulate less Hamming weight to be corrected) than the codes with higher $T_{d_{free}}({\cal C})$ which demand more separation in network uses between error events for them to be corrected (despite having a greater flexibility in the Hamming weight accumulated by the correctable error events).
\end{itemize} 
\begin{remark}
The difference in the performance of code ${\cal C}_1$ between Sink 1 and Sink 2 is probably due to the unequal error protection to the two code symbols. When the code is `reversed' ,i.e. with $G_{I_1}(z)~=~[1~~~1+z]$, it is observed that the performance at the sinks are also interchanged for unchanged error characteristics. 
\end{remark}
\section{Concluding remarks}
\label{sec8}
In this work, we have extended the approach of \cite{PrR} to introduce network error correction for acyclic, unit-delay, memory-free networks. A construction of CNECCs for acyclic, unit-delay, memory-free networks has been given, which corrects errors corresponding to a given set of patterns as long as consecutive errors are separated by a certain number of network uses. Bounds are derived on the field size required for the construction of a CNECC with the required error correction capability and also on the minimum separation in network uses between any two consecutive network errors. Simulations assuming a probabilistic error model on a modified butterfly network indicate the implementability and performance tractability of such CNECCs. The following problems remain to be investigated.
\begin{itemize}
\item Investigation of error correction bounds for network error correction in unit-delay, memory-free networks.
\item Joint design of the CNECC and network code.
\item Investigation of distance bounds for CNECCs.
\item Design of appropriate processing matrices at the sinks to minimize the maximum Hamming weight of the error sequences.
\item Construction of CNECCs which are optimal in some sense.
\item Further analytical studies on the performance of CNECCs on unit-delay networks.
\end{itemize}
\section*{Acknowledgment} This work was supported  partly by the DRDO-IISc program on Advanced Research in Mathematical Engineering to B.~S.~Rajan.

\appendices%
\section{Convolutional codes-Basic Results}
\label{app1}
We review the basic concepts related to convolutional codes, used extensively throughout the rest of the paper. For $q,$ power of a prime, let $\mathbb{F}_q$ denote the finite field with $q$ elements,  $\mathbb{F}_q[z]$ denote \textit{the ring of univariate polynomials} in $z$ with coefficients from $\mathbb{F}_q,$  $\mathbb{F}_q(z)$ denote \textit{the field of rational functions} with variable $z$ and coefficients from $\mathbb{F}_q$ and $\mathbb{F}_q[[z]]$ denote  \textit{the ring of formal power series} with coefficients from $\mathbb{F}_q$. Every element of $\mathbb{F}_q[[z]]$ of the form $x(z)=\sum_{i=0}^\infty x_iz^i, x_i \in \mathbb{F}_q$. Thus, $\mathbb{F}_q[z] \subset \mathbb{F}_q[[z]]$. We denote the set of $n$-tuples over $\mathbb{F}_q[[z]]$ as $\mathbb{F}_q^n[[z]]$. Also, a rational function $x(z)=\frac{a(z)}{b(z)}$ with $b(0) \neq 0$ is said to be \textit{realizable}. A matrix populated entirely with realizable functions is called a realizable matrix.

For a convolutional code, the \textit{information sequence} $\boldsymbol{u} = \left[\boldsymbol{u}_0,\boldsymbol{u}_1,...,\boldsymbol{u}_t\right](\boldsymbol{u}_i\in\mathbb{F}_q^b)$ and the \textit{codeword sequence} (output sequence) $\boldsymbol{v} = \left[\boldsymbol{v}_0,\boldsymbol{v}_1,...,\boldsymbol{v}_t\right]\left(\boldsymbol{v}_i\in\mathbb{F}_q^c\right)$ can be represented in terms of the delay parameter $z$ as 	
\begin{eqnarray*}
\boldsymbol{u}(z)=\sum_{i=0}^t \boldsymbol{u}_i z^i ~~~ \mbox{  and  }~~~
\boldsymbol{v}(z)=\sum_{i=0}^t \boldsymbol{v}_i z^i
\end{eqnarray*}
\begin{definition}[\cite{JoZ}]
A \textit{convolutional code}, ${\cal C}$ of rate $~b/c~(b~<~c)$ is defined as 
\[
{\cal C} = \{ \boldsymbol{v}(z)\in\mathbb{F}_q^{c}[[z]]~|~ \boldsymbol{v}(z)=\boldsymbol{u}(z)G(z) \}
\] 
where $G(z)$ is a $b \times c$  \textit{generator matrix} with entries from $\mathbb{F}_q(z)$ and rank $b$ over $\mathbb{F}_q(z)$, and $\boldsymbol{v}(z)$ being the codeword sequence arising from the information sequence, $\boldsymbol{u}(z)\in\mathbb{F}_q^{b}[[z]]$.
\end{definition}

Two generator matrices are said to be \textit{equivalent} if they encode the same convolutional code. A \textit{polynomial generator matrix}\cite{JoZ} for a convolutional code $\cal C$ is a generator matrix for $\cal C$ with all its entries from $\mathbb{F}_q[z]$. It is known that every convolutional code has a polynomial generator matrix \cite{JoZ}. Also, a generator matrix for a convolutional code is \textit{catastrophic}\cite{JoZ} if there exists an information sequence with infinitely many non-zero components, that results in a codeword with only finitely many non-zero components.

For a polynomial generator matrix $G(z)$, let $g_{ij}(z)$ be the element of $G(z)$ in the $i^{th}$ row and the $j^{th}$ column, and 
\[
\nu_i:=\max_{j} deg(g_{ij}(z))
\]
be the $i^{th}$ \textit{row degree} of $G(z)$. Let 
\[
\delta: = \sum_{i=1}^{b}\nu_i
\]
be the \textit{degree} of $G(z).$
\begin{definition}[\cite{JoZ} ]
A polynomial generator matrix is called \textit{basic} if it has a polynomial right inverse. It is called \textit{minimal} if its degree $\delta$ is minimum among all generator matrices of $\cal C$.
\end{definition}

Forney in \cite{For} showed that the ordered set $\left\{\nu_{1},\nu_{2},...,\nu_{b}\right\}$ of row degrees (indices) is the same for all minimal basic generator matrices of $\cal C$ (which are all equivalent to one another). Therefore the ordered row degrees and the degree $\delta$ can be defined for a convolutional code $\cal C.$ A rate $b/c$ convolutional code with degree $\delta$ will henceforth be referred to as a $(c,b,\delta)$ code. Also, any minimal basic generator matrix for a convolutional code is non-catastrophic. 

\begin{definition}[\cite{JoZ} ]
A \textit{convolutional encoder} is a physical realization of a generator matrix by a linear sequential circuit. Two encoders are said to be \textit{equivalent encoders} if they encode the same code. A \textit{minimal encoder} is an encoder with the minimal number of delay elements among all equivalent encoders.
\end{definition}

The weight of a vector $\boldsymbol{v}(z) \in \mathbb{F}_q^{c}[[z]]$ is the sum of the Hamming weights (over $\mathbb{F}_q$) of all its $\mathbb{F}_q^{c}$-coefficients. Then we have the following definitions.
\begin{definition}[\cite{JoZ}]
The \textit{free distance} of a convolutional code $\cal C$ is given as 
\[
d_{free}({\cal C})=min\left\{wt(\boldsymbol{v}(z))|\boldsymbol{v}(z)\in{\cal C},\boldsymbol{v}(z)\neq 0\right\}
\]
\end{definition}
\section{MDS convolutional codes}
\label{app2}
We discuss some results on the existence and construction of Maximum Distance Separable (MDS) convolutional codes.

The following bound on the free distance, and the existence of codes meeting the bound, called MDS convolutional codes, was proved in \cite{RoS}. 
\begin{theorem}[\cite{RoS}] 
\label{GenSingBound}
For every base field $\mathbb{F}$ and every rate $k/n$ convolutional code $\cal C$ of degree $\delta$, the free distance is bounded as 
\[
d_{free}({\cal C})\leq(n-k)(\lfloor \delta / k \rfloor + 1) + \delta + 1
\]
\end{theorem}
Theorem \ref{GenSingBound} is known as the \textit{generalized Singleton bound}.	
\begin{theorem}[\cite{RoS}]For any positive integers $k<n$, $\delta$ and for any prime $p$ there exists a field $\mathbb{F}_q$ of characteristic $p$, and a rate $k/n$ convolutional code $\cal C$ of degree $\delta$ over $\mathbb{F}_q$, whose free distance meets the generalized Singleton bound.
\end{theorem}

A method of constructing MDS convolutional codes based on the connection between quasi-cyclic codes and convolutional codes was given in \cite{RLS}. The ordered Forney indices for such codes are of the form 
\[
\nu_{1}=\nu_{2}=...=\nu_{l} < \nu_{l+1} =...=\nu_{k}.
\]
where $\nu_{1}=\lfloor \delta/k \rfloor$ and $\nu_{k}=\lfloor \delta/k \rfloor + 1.$

It is known \cite{RLS} that the field size $q$ required for a $(n,k,\delta)$ convolutional code ${\cal C}$ with $d_{free}({\cal C})$ meeting the generalized Singleton bound in the construction in \cite{RLS} needs to be a prime power such that 
\begin{equation}
\label{fieldsizeconv}
n|(q-1)\text{ and } q\geq\delta\frac{n^2}{k(n-k)}+2.
\end{equation}

\end{document}